\pdfoutput=1
\documentclass[11pt,letter]{article}
\usepackage{amsmath,amssymb,amsthm,amsfonts}
\usepackage{fullpage}
\usepackage{url}
\usepackage{color}
\usepackage{algpseudocode}
\usepackage{algorithm}
\usepackage{graphicx}
\usepackage{cleveref}
\usepackage{arydshln}
\usepackage{verbatim}
\usepackage{enumitem}

\usepackage{times}

\DeclareMathOperator{\poly}{poly}
\DeclareMathOperator{\supp}{supp}

\DeclareMathAlphabet{\pazocal}{OMS}{zplm}{m}{n}

\newcommand{\Oh}{\mathcal{O}}

\newcommand{\R}{\mathbb{R}}

\makeatletter
\newtheorem*{rep@theorem}{\rep@title}
\newcommand{\newreptheorem}[2]{%
\newenvironment{rep#1}[1]{%
 \def\rep@title{#2 \ref{##1}}%
 \begin{rep@theorem}}%
 {\end{rep@theorem}}}
\makeatother

\newtheorem{theorem}{Theorem}
\newreptheorem{theorem}{Theorem}
\newtheorem{lemma}[theorem]{Lemma}

\theoremstyle{definition}
\newtheorem{definition}{Definition}

\crefname{remark}{Remark}{Remarks}

\author{Yi Li\\
			Nanyang Technological University\\
			\texttt{\small yili@ntu.edu.sg}
	\and
			Vasileios Nakos\thanks{Supported in part by NSF grant IIS-1447471.}\\
			Harvard University\\
			\texttt{\small vasileiosnakos@g.harvard.edu}
} 

\title{Deterministic Heavy Hitters with Sublinear Query Time}
\date{}
\begin{document}
\maketitle

\begin{abstract}

We study the classic problem of finding $\ell_1$ heavy hitters in the streaming model. In the general turnstile model, we give the first deterministic sublinear-time sketching algorithm which takes a linear sketch of length $\Oh(\epsilon^{-2} \log n \cdot \log^*(\epsilon^{-1}))$, which is only a factor of $\log^*(\epsilon^{-1})$ more than the best existing polynomial-time sketching algorithm (Nelson et al., RANDOM '12). Our approach is based on an iterative procedure, where most unrecovered heavy hitters are identified in each iteration. Although this technique has been extensively employed in the related problem of sparse recovery, this is the first time, to the best of our knowledge, that it has been used in the context of heavy hitters. Along the way we also obtain a sublinear time algorithm for the closely related problem of the $\ell_1/\ell_1$ compressed sensing, matching the space usage of previous (super-)linear time algorithms. In the strict turnstile model, we show that the runtime can be improved and the sketching matrix can be made strongly explicit with $O(\epsilon^{-2}\log^3 n/\log^3(1/\epsilon))$ rows.

\end{abstract}


\section{Introduction}

The problem of detecting \emph{heavy hitters}, also frequently referred to as \emph{elephants} or \emph{hot items}, is one of the most well-studied problems in databases and data streams, from both theoretical and practical perspectives. In this problem, we are given a long data stream of elements coming from a large universe, and we are asked to report all the elements that appear at least a large number of times (called \emph{heavy hitters}), using space that is much smaller than the size of the universe and the length of the stream. 
	
Finding popular terms in search queries, identifying destination adresses of packets, detecting anomalies in network traffic streams such as denial-of-service (DoS) attacks, or performing traffic engineering, are only some of the important practical appearances of the heavy hitters problem. For example, the central task of managing large-scale networks lies in accurately measuring and monitoring network traffic \cite{zhang2010identifying,yang2016heavy}. Interestingly, empirical studies \cite{fred2001statistical,mori2004characteristics,papagiannaki2001feasibility,zhang2002characteristics} indicate that flow-statistics in large networks follow an elephant/mice phenomenon, i.e.,  the vast majority of the bytes are concentrated on only a small fraction of the flows.

On the theoretical side, heavy hitters appear very often, both in streaming algorithms and sparse recovery tasks. For the problems such as streaming entropy estimation~\cite{harvey:entropy}, $\ell_p$ sampling~\cite{precisionsampling,monemizadeh:sampling}, finding duplicates~\cite{jowhari:sampling}, block heavy hitters~\cite{JW:cascaded}, sparse recovery tasks~\cite{GSTV07,GLPS12,GLPS17}, many algorithmic solutions use heavy hitters algorithms as subroutines. 

\medskip
\noindent \textbf{Streaming Models.} In this paper we consider the most general streaming model, called the \emph{(general) turnstile model}, defined as follows. There is an underlying vector $x\in \R^n$, which is initialized to zero and is maintained throughout the input stream. Each element in the input stream describes an update $x_i\gets x_i +\delta$ for some index $i$ and increment $\delta$, where $\delta$ can be either positive or negative.

We also consider a restricted version of the general turnstile model, called the \emph{strict turnstile model}, under which it is guaranteed that $x_i\geq 0$ for all $i$ throughout the input stream. This restricted model captures the practical scenario where item deletions are allowed but an item cannot be deleted more than it is inserted.

\medskip
\noindent\textbf{Sketching Algorithms.} An important class of streaming algorithms are called \emph{sketching algorithms}. A sketching algorithm maintains a short linear sketch $v = \Phi x$ (where $\Phi\in \R^{m\times n}$) throughout the input stream and then runs a recovery algorithm $\mathcal{D}$, which has access to only $v$ and $\Phi$, to output a desired $\hat x$. The space usage is proportional to $m$ (the length of the sketch $v$) and to the memory needed to store $\Phi$. Therefore we wish to minimize $m$ and design a structured $\Phi$ such that storing $\Phi$ takes little space. 
Surprisingly, all existing streaming algorithms under the general turnstile model are sketching algorithms, and it has been shown~\cite{LNW14stoc} that all streaming algorithms under the general turnstile model can be converted to sketching algorithms with a mild increase in the space usage.

\subsection{$\ell_\infty/\ell_1$ Sparse Recovery}
The heavy hitter problem has been studied under various streaming models and various recovery guarantees. Depending on the heaviness we are interested in, we distinguish between $\ell_1$ and $\ell_2$ heavy hitters. We are interested in finding, in the first case, the coordinates which are at least $\epsilon\|x\|_1$ in magnitude, and in the second case, the coordinates which are at least $\epsilon\|x\|_2$. Although finding $\ell_2$ heavy hitters is strictly stronger than finding $\ell_1$ heavy hitters, we consider only $\ell_1$ heavy hitters in this paper, for it is impossible to find $\ell_2$ heavy hitters using a deterministic space-saving sketching algorithm (see details below).

Specifically, we consider a classical recovery guarantee, called the $\ell_\infty/\ell_1$ error guarantee in the literature, that is, the algorithm outputs an $\Oh(1/\epsilon)$-sparse vector $\hat x$ such that 
\begin{equation}\label{eqn:ellinf/ell1}
\|\hat x-x\|_\infty \leq \epsilon\|x_{-r}\|_1,
\end{equation}
for some parameters $\epsilon$ and $r$, where $x_{-r}$ is the vector obtained by zeroing out the largest $r$ coordinates of $x$ in magnitude (absolute value). This type of guarantee requires not only finding the heavy hitters, but also giving `good enough' estimates of them, where the estimates are measured with respect to $\|x_{-r}\|_1$ instead of the larger $\|x\|_1$; this type of guarantee is called the \emph{tail guarantee}. It should be noted that the $\ell_{\infty}/\ell_1$ guarantee has been extensively studied and is provided by several classical algorithms, e.g. \textsc{Count-Min}~\cite{cormode2005improved}, \textsc{LossyCounting}~\cite{lossycounting}, \textsc{SpaceSaving}~\cite{spacesaving}, although not all of them work under the general turnstile model.

In this paper we focus on deterministic sketching algorithms, which means that both the matrix $\Phi$ and the recovery algorithm $\mathcal{D}$ allow uniform reconstruction of every $x \in \mathbb{R}^n$ up to $\epsilon \|x_{-r}\|_1$ error, providing the best applicability. This is also referred to as ``for-all'' guarantee in the sparse recovery literature, in contrast to the ``for-each'' guarantee, which allows reconstruction of a fixed vector with some target success probability. Most previous sketching algorithms for the heavy hitter problems concern the ``for-each'' model and resort to randomization (e.g.~\cite{cormode2005improved,BCIS10}), by drawing a random $\Phi$ from some distribution and guaranteeing $\mathcal{D}$ to output an acceptable $\hat x$ with a good probability. Other sketching algorithms are deterministic, however, they run in time at least linear in the universe size $n$. The goal of fast query time, say, logarithmic in $n$, is crucial to streaming applications. For instance, in traffic monitoring $n$ equals the number of all possible packets, namely $2^{32}$; a linear runtime would be prohibitive in any reasonable real-world scenario. A natural goal is to design a sketching algorithm with sublinear query time, preferably $\Oh(\poly(1/\epsilon,r,\log n))$, with as little space usage as possible.

Apart from practical importance, deterministic algorithms for heavy hitters is an interesting theoretical subfield of streaming algorithms, connected to dimensionality reduction and incoherent matrices \cite{nnw12}. Moreover, gaining insight into such questions may give insight for many other data stream problems where heavy hitter algorithms are used as subroutines. We note that any deterministic sketching algorithm that finds $\ell_2$ heavy hitters requires $\Omega(n)$ space \cite{cohen2009compressed} (which implies that the trivial algorithm storing the entire input vector $x$ is asymptotically optimal), while the best lower bound for $\ell_1$ heavy hitters is $\Omega( r \log(n/r)/\log r + \epsilon^{-2} + \epsilon^{-1} \log n)$ \cite{nnw12,ganguly2008lower}.

The state-of-the-art deterministic sketching algorithms for $\ell_1$ heavy hitters are found in~\cite{nnw12}, where two algorithms are given. The first algorithm uses $m = \Oh(\epsilon^{-2} \log n \cdot \min\{1, \log n/(\log \log n + \log(1/\epsilon))^2\})$ rows and achieves $\epsilon \|x_{-\lceil 1/\epsilon\rceil}\|_1$ tail guarantee (setting $r = \lceil 1/\epsilon\rceil$ in \eqref{eqn:ellinf/ell1}). The second algorithm uses $m = \Oh(\epsilon^{-2}\log n)$ rows and  achieves a stronger $\epsilon\|x_{-\lceil 1/\epsilon^2\rceil}\|_1$ tail guarantee. We can see that when $\epsilon < 2^{-\Omega(\sqrt{\log n})}$, the first algorithm uses less space, whereas when $\epsilon \geq 2^{-\Omega(\sqrt{\log n})}$ the second algorithm is better. The number of rows used by both algorithms is at most suboptimal by a $\log n$ factor while the runtimes are both superlinear $\Omega(\epsilon^{-1} n  \log n)$. In this paper our goal is to obtain an algorithm which runs in sublinear time in $n$ while attaining the stronger $\epsilon\|x_{-\lceil 1/\epsilon^2\rceil}\|_1$ tail guarantee with near-optimal number of rows. Our main theorem is formally stated below.

\begin{theorem}[$\ell_{\infty}/\ell_1$]\label{thm:ell_infty/ell_1}
There exists a linear sketch $\Phi \in \mathbb{R}^{m \times n}$ such that for every $x\in\R^n$, we can, given $\Phi x$, find an $\Oh(1/\epsilon)$-sparse vector $\hat{x}$ such that
\[	
	\|x-\hat{x}\|_{\infty} \leq \epsilon\|x_{-\lceil 1/\epsilon^2 \rceil}\|_1,	
\]
in $\Oh((1/\epsilon)^6 \poly(\log n))$ time. The number of rows of $\Phi$ equals $m = \Oh(\epsilon^{-2} \log n\log^\ast (\epsilon^{-1}))$.
\end{theorem}

The number of rows in $\Phi$ is more than that in~\cite{nnw12} by merely a factor of $\Oh(\log^\ast (1/\epsilon))$, while the query time is sublinear for all small $\epsilon\geq n^{-1/7}$, a significant improvement upon the previous $\Oh( \epsilon^{-1} n \log n)$ runtime in~\cite{nnw12}. 

\medskip
\noindent\textbf{Difference Between Deterministic and Explicit Schemes} To avoid confusion, we note the difference between `deterministic' and `explicit'. In the compressed sensing/sparse recovery literature `deterministic' is also called `for-all' or `uniform', which means that a single matrix $\Phi$ suffices for the reconstruction of all vectors. `Explicit' means that the matrix $\Phi$ can be constructed in time $\poly(1/\epsilon, r, n)$. `Strongly explicit' means that any entry of the matrix can be computed in time $\poly(1/\epsilon,r,\log n)$. Hence, in this paper we show the existence of a single matrix that allows reconstrction of all vectors, which we argue via the probabilistic method. The same holds for some of the schemes in \cite{nnw12}. 

In the strict turnstile model, we show that for the tail guarantee of $r=\lceil 1/\epsilon\rceil$, the matrix $\Phi$ can be made strongly explicit, with a mild increase in the number of rows, and the runtime can be improved polynomially. Note, however, the tail guarantee is with respect to $r=\lceil 1/\epsilon\rceil$ instead of $r = \lceil 1/\epsilon^2\rceil$. We state our theorem below and shall prove it in \Cref{sec:strictturnstile}.

\begin{theorem}[$\ell_\infty/\ell_1$, strict turnstile]\label{thm:strict_turnstile}
There exists a strongly explicit matrix $M$ of $O((1/\epsilon)^2 \log^3 n/\log^3(1/\epsilon))$ rows, which, given $Mx$ in the strict turnstile model, allows us to find an $O(1/\epsilon)$-sparse vector $\hat x$ such that
	\[
		\|x - \hat x\|_{\infty} \leq \epsilon \|x_{-\lceil 1/\epsilon\rceil}\|_1,
	\]
in time $O((1/\epsilon)^3 \log^3 n/\log^3(1/\epsilon))$. 
\end{theorem}

\subsection{$\ell_1/\ell_1$ Sparse Recovery}
In the $\ell_1/\ell_1$ sparse recovery problem, instead of the guarantee \eqref{eqn:ellinf/ell1}, the algorithm should output $\hat x$ such that 
\begin{equation}\label{eqn:ell1/ell1}
\|\hat x- x\|_1\leq (1+\epsilon)\|x_{-k}\|_1.
\end{equation}

It is known that any deterministic sketching algorithm requires $m = \Omega(\epsilon^{-2} + \epsilon^{-1}k\log(\epsilon n/k))$ rows of $\Phi$~\cite{nnw12}, and the best known upper bound is $m = O(\epsilon^{-2}k\log(n/k))$ rows~\cite{indyk:ell_1,indyk:unified,GLPS17}, suboptimal from the lower bound by only a logarithmic factor. However all these algorithms suffer from various defects: the algorithms in \cite{indyk:ell_1,indyk:unified} run in polynomial time in $n$, and that in~\cite{GLPS17} imposes a constraint on $\epsilon$ that precludes it from being a small constant when $k$ is small. In this paper, we show that one can achieve sublinear runtime with the same number of rows for small $k$.

\begin{theorem}[$\ell_1/\ell_1$]\label{thm:ell_1/ell_1}
There exists a matrix $\Phi \in \mathbb{R}^{m \times n}$ such that, given $\Phi x$ with $x\in\mathbb{R}^n$, we can find an $\Oh(k)$-sparse vector $\hat{x}$ satisfying \eqref{eqn:ell1/ell1} in $\Oh(k^{3} \poly(1/\epsilon, \log n))$ time. The number of rows of $\Phi$ is $m = \Oh(\epsilon^{-2}k \log n)$. 
\end{theorem}

This result is the first sublinear time algorithm for all small $k\leq n^{0.3}$ with constant $\epsilon$, while the algorithm in~\cite{GLPS17}, using the same number of measurements, works for constant $\epsilon$ only when $n^{\delta'}\leq k\leq n^{1-\delta}$ (where $\delta, \delta' > 0$ are arbitrarily small constants), owing to its use of the list-decodable code. Combining the two results, we have solved the $\ell_1/\ell_1$ problem in sublinear time for all $k\leq n^{1-\delta}$ and constant $\epsilon$.

\smallskip
\noindent\textbf{Remark } Our results heavily involve random hash functions, for which $\Oh(1/\epsilon)$- or $\Oh(k)$-wise independence would be sufficient. The space complexity of our algorithms is the same as the number of rows, unless stated otherwise.


\section{Overview of Techniques}\label{sec:techniques}

The main result on $\ell_{\infty}/\ell_1$ combines different ideas from sparse recovery and heavy hitters literature. We first prove a result with the $\epsilon\|x_{-1/\epsilon}\|_1$ tail guarantee. We need a different, more careful construction of the weak system, akin to that in \cite{GLPS17}, which does not detect only a constant fraction of the heavy hitters, but a much larger fraction, as much as $(1- \epsilon \log \log (1/\epsilon) )$. One of our technical contributions and tools is the design of a more general form of the weak system, which we then apply iteratively with carefully chosen parameters to recover all heavy hitters. We then iterate by subtracting the found heavy hitters, and try to find the remaining ones using a new matrix of the same number of resources (rows). Similar iteration techniques have been adopted in most combinatorial sparse recovery tasks, where the algorithms are allowed to miss even all heavy hitters if they are not large enough! Our $\ell_\infty/\ell_1$ error guarantee, however, makes it prohibitive to miss small heavy hitters in later iterations. To that end, we heavily exploit the more abundant resource of $1/\epsilon^2$ rows in each iteration throughout, for which we pay a mild extra factor in the total number of rows. While in previous sparse recovery tasks the number of rows decreases across different iterations, this does not happen in our case. As aforementioned, we pay an additional $\log^{\ast}(1/\epsilon)$ factor as we shall have $\Oh(\log^{\ast}(1/\epsilon))$ iterations until all heavy hitters are recovered. 

To obtain the stronger tail guarantee of $\epsilon\|x_{-1/\epsilon^2}\|_1$, we invoke additionally our $\ell_1/\ell_1$ algorithm and the point-query algorithm of \cite{nnw12}. We note that any sub-optimality in the number of rows of the $\ell_1/\ell_1$ linear sketch would yield a worse result for our main scheme, which forces us to obtain also an improved result for the $\ell_1/\ell_1$ problem. Our new weak system and the novel idea of using the iterative loop to satisfy the $\ell_{\infty}/\ell_1$ guarantee may indicate new approaches to tackle heavy hitters tasks, and might be of interest beyond the scope of this paper. Our side-result on $\ell_1/\ell_1$ sparse recovery, is a combination of \cite{GLPS17} and \cite{LNNT}. Specifically, one can avoid the Parvaresh-Vardy list-recoverable code that \cite{GLPS17} employed, and use instead the clustering technique in \cite{LNNT}, upon the two-layer hashing schemes and linking technique in~\cite{GLPS17}. This makes possible an improved result for $\ell_1/\ell_1$ that removes the restriction on $\epsilon$ the previous work of \cite{GLPS17} was suffering from. 

	We remark that any improvement in the running time of the clustering algorithm of \cite{LNNT} immediately translates to improvement to $\ell_{\infty}/\ell_1$ and $\ell_1/\ell_1$ schemes. More specifically, a near-quadratic or near-linear algorithm for that clustering would imply a near-quadratic or near-linear (in the number of rows of the sketching matrix) time algorithm for all three of our tasks. The current state of the art for that algorithm is $\tilde{\Oh}(N^3)$ runtime on a graph of $N$ vertices, since the algorithm performs $N$ calls to a routine that finds a Cheeger cut. We also remark that we could obtain our $\ell_{\infty}/\ell_1$ result using explicit list-recoverable codes, such as Parvaresh-Vardy code, but this would lead to a slightly worse result than what we currently have. 

In the strict turnstile model, we obtain a family of strongly explicit matrices that solves the point query problem in sublinear time using the family of matrices from~\cite{nnw12}. First, we show how to stengthen the guarantee obtained by the matrix in~\cite{nnw12}, achieving a stronger tail guarantee, and then we show how to recursively combine those explicit matrices to obtain a sketching algorithm that gives the $\ell_{\infty}/\ell_1$ guarantee in sublinear time. Our result in the strict turnstile model, not only is strongly explicit, but also has a better running time than the general turnstile model and gets also improved space for some regime of $\epsilon$, namely $\epsilon \leq 2^{-\sqrt{\log n}}$. We note though that we obtain a weaker tail guarantee than in the general case, namely $ r= 1/\epsilon$, instead of $1/\epsilon^2$. This weaker guarantee stems from the lack of strongly explicit $\ell_1/\ell_1$ schemes with nearly optimal measurements, that can also answer queries in sublinear time.

\section{Preliminaries}\label{sec:preli}

For a vector $x \in \mathbb{R}^n$, we define $x_{-k}$ to be the vector obtained by zeroing out the largest $k$ coordinates in magnitude, and $\operatorname{supp}(x)$ to be the set of the non-zero coordinates of $x$. We also define $H(x,k)$ to be the index set of the largest $k$ coordinates of $x$ in magnitude. Thus, $\operatorname{supp}(x_{-k}) \cap H(x,k)= \emptyset$. We also define $H(x,k,\epsilon) = \{ i \in [n]: |x_i| \geq \epsilon/ k \|x_{-k}\|_1 \}$. We assume that the word size is $ w = \Theta(\log n)$. 

An error correcting code is a subset $\mathcal{C}\subseteq \Sigma^n$, where $\Sigma$ is a finite set called alphabet and $|\mathcal{C}| = |\Sigma|^k$ for some $k\geq n$, together with an injective encoding map $\operatorname{enc}:\Sigma^k\to\mathcal{C}$ and a decoding map $\operatorname{dec}:\mathcal{C}\to \Sigma^k$. The parameter $k$ is called the message length and $n$ is called codeword length or block length. We say that an error correcting code can correct up to $\theta$-fraction of errors if for any message $m\in\Sigma^k$ and any $x\in \Sigma^n$ such that the Hamming distance $d(\operatorname{enc}(m),x)\leq \theta n$, it holds that $\operatorname{dec}(x) = m$.

\subsection{Two-layer Hashing Schemes}\label{sec:skeleton}

In this subsection we review the two-layer hashing scheme and the linking techniques used in~\cite{GLPS17}, which will be the skeleton of construction for all our sparse recovery results. 

First we recall some definitions, taken from~\cite{GLPS17}, regarding bipartite expander, two-layer hashing and isolation of heavy hitters.

\begin{definition}[bipartite expander]
An $(n,m,d,\ell,\epsilon)$-bipartite expander is a $d$-left-regular bipartite graph $G(L\cup R, E)$ where $|L| = n$ and $|R| = m$ such that for any $S\subseteq L$ with $|S|\leq \ell$ it holds that $|\Gamma(S)|\geq (1-\epsilon)d|S|$, where $\Gamma(S)$ is the neighbour of $S$ (in $R$). When $n$ and $m$ are clear from the context, we abbreviate the expander as $(\ell,d,\epsilon)$-bipartite expander.
\end{definition}

\begin{definition}[one-layer hashing scheme]
The $(N,B,d)$ (one layer) hashing scheme is the uniform distribution on the set of all functions $f:[N]\to [B]^d$. We write $f(x) = (f_1(x),\dots,f_d(x))$, where $f_i$'s are independent $(N,B)$ hashing schemes.
\end{definition}
Each instance of such a hashing scheme induces a $d$-left-regular bipartite graph with $Bd$ right nodes. When $N$ is clear from the context, we simply write $(B,d)$ hashing scheme.

\begin{definition}[two-layer hashing scheme]
An $(N,B_1,d_1,B_2,d_2)$ (two-layer) hashing scheme is a distribution $\mu$ on the set of all functions $f:[N]\to [B_2]^{d_1d_2}$ defined as follows. Let $g$ be a random function subject to the $(N,B_1,d_1)$ hashing scheme and $\{h_{i,j}\}_{i\in[d_1],j\in[d_2]}$ be a family of independent functions subject to the $(B_1,B_2,d_2)$ hashing scheme which are also independent of $g$. Then  $\mu$ is defined to be the distribution induced by the mapping
\begin{multline*}
x\mapsto \left(h_{1,1}(g_1(x)),\dots,h_{1,d_2}(g_1(x)),h_{2,1}(g_2(x)),\dots,h_{2,d_2}(g_2(x)),\dots,\right.\\
\left.h_{d_1,1}(g_{d_1}(x)),\dots,h_{d_1,d_2}(g_{d_1}(x))\right).
\end{multline*}
\end{definition}
Each instance of such a hashing scheme gives a $d_1d_2$-left-regular bipartite graph of $B_2 d_1 d_2$ right nodes. When $N$ is clear from the context, we simply write $(B_1,d_1,B_2,d_2)$ hashing scheme. Conceptually we hash $N$ elements into $B_1$ buckets and repeat $d_1$ times; these buckets will be referred to as first-layer buckets. In each of the $d_1$ repetitions, we hash $B_1$ elements into $B_2$ buckets and repeat $d_2$ times, those buckets will be referred to as second-layer buckets.

Bipartite expander graphs can be used as hashing schemes because of their isolation property.

\begin{definition}[isolation property]
An $(n,m,d,\ell,\epsilon)$-bipartite expander $G$ is said to satisfy the $(\ell, \eta, \zeta)$-isolation property if for any set $S\subset L(G)$ with $|S|\leq \ell$, there exists $S'\subset S$ with $|S'|\geq (1-\eta)|S|$ such that for all $x\in S'$ it holds that $|\Gamma(\{x\})\setminus\Gamma(S\setminus\{x\})|\geq (1-\zeta) d$.
\end{definition}

The following lemma shows that a random two-layer hashing satisfies a good isolation property with high probability. Previous works~\cite{PS12,GLPS17} build sparse recovery systems upon this lemma.

\begin{lemma}[{\cite{GLPS17}}]\label{lem:two-layer-isolation}
Let $\epsilon > 0$, $\alpha>1$ and $(N, B_1,d_1,B_2,d_2)$ be a two-layer hashing scheme with $B_1=\Omega(\frac{k}{\zeta^\alpha\epsilon^{2\alpha}})$, $d_1=\Omega(\frac{\alpha}{\alpha-1}\cdot \frac{1}{\zeta\epsilon}\frac{\log N}{\log(B_1/k)})$, $B_2 = \Omega(\frac{k}{\zeta\epsilon})$ and $d_2 = \Omega(\frac{1}{\zeta}\log\frac{B_1}{k})$. Then with probability $\geq 1-1/N^c$, the two-layer hashing scheme with parameters prescribed above gives an $(N, B_2d_1d_2, d_1d_2, 4k, \epsilon)$ bipartite graph with the $(L, \epsilon, \zeta)$-isolation property, where $L=O(k/\epsilon)$.
\end{lemma}

\subsection{Message Encoding}
We give a brief review of the construction and message encoding in \cite{GLPS17}. Let $\operatorname{enc}: \{0,1\}^{\log n} \rightarrow \{0,1\}^{\Oh(\log n)}$ be an error-correcting code that corrects a constant fraction of errors in linear time. For notational convenience, let $m_i = \mathrm{enc}(i)$, the codeword for the binary representation of $i$. 
Furthermore, we break $m_i$ into $d_1$ blocks of length $\Theta((\log n)/d_1)$ each, say, $m_i = (m_{i,1},\dots,m_{i,d_1})$.

Let $G$ be a $\Delta$-regular edge-expander graph on $d_1$ vertices, where $\Delta$ is an absolute constant (we may assume that $d_1$ is even and such edge expander exists by~\cite{pinsker}). Let $j$ be a node in $G$ and denote its neighbours by $\Gamma_1(j),\dots,\Gamma_\Delta(j)$. Let $idx(r, i)$ ($r\in [d_1]$ and $i\in [n]$) denote the index of the bucket where $i$ is hashed in the $r$-th first-layer repetition. Construct the message
\[
\bar m_{i,r} = m_{i,r} \circ idx(\Gamma_1(r), i) \circ \cdots \circ idx(\Gamma_\Delta(r), i),\qquad i\in [n], r\in [d_1],
\]
where $\circ$ denotes concatenation of strings and $idx(\cdot,\cdot)$ is understood as its binary representation of $\log(B_1)$ bits. 

Now for each index $i$ we have $d_1$ blocks of message $\bar m_{i,1},\dots,\bar m_{i,d_1}$. We can protect each block using a constant-rate error correcting code which tolerates a constant fraction of error and decodes in polynomial time, so that if we can recover a fraction of $\bar m_{i,r}$ we can recover the entire message $\bar m_{i,r}$ efficiently. 
The high-level idea is then to recover a good fraction of $\{\bar m_{i,r}\}_{r\in [d]}$ for a good fraction of heavy hitters $i$, so that we can recover $m_i$ using the linking information embedded in $\bar m_{i,r}$ and the clustering algorithm in~\cite{LNNT}. Finally we decode $m_i$ to obtain the corresponding index $i$.

The following lemma is crucial in bounding the number of missed heavy hitters in iteration, modified from {\cite{GLPS17}}. The proof is similar to that in \cite{GLPS17}.

\begin{lemma}\label{lem:decoy}
Let $\theta, \epsilon\in (0,1)$, $\delta\in(0,\frac{1}{2}]$ and $\beta,\zeta > 0$ such that $0 < \zeta < \delta - \frac{64\beta}{\theta}$.
Suppose that $G$ is a $(4s,d,\beta \epsilon)$-bipartite expander which satisfies the $(\frac{6}{\gamma\epsilon},\frac{\epsilon\theta}{12},\zeta)$-isolation property, where $\gamma\in [\frac{\theta}{s},1]$. Let $x\in \mathbb{R}^n$ be a vector which can be written as $x = y+z$, where $y$ and $z$ have disjoint supports, $|\supp(y)|\leq s$ and $\|z\|_1\leq 3/2$. For each $i\in [n]$ define the multiset $E_i$ as
\[
E_i = \left\{ \sum_{(u,v)\in E} x_u \right\}_{v\in \Gamma(\{i\})}.
\]
Note that $|E_i| = d$ since it is a multiset. Then, for every $D \subset [n], |D| \leq 2s$, we have that \[
\left|\left\{i\in D: |x_i - w| \geq \frac{\epsilon\gamma}{4}\text{ for at least }(1-\delta)d\text{ values }w\text{ in }E_i\right\}\right| \leq \frac{\theta}{\gamma}.
\]
\end{lemma}
\begin{proof}
Suppose that $|D| > \theta s$, otherwise the result holds automatically. 
Assume that $|x_1|\geq |x_2|\geq \cdots \geq |x_n|$. 
Let $T = D\cup \{i: |x_i|\geq \epsilon\gamma/4\}$, then $t := |T|\leq \|z\|_1/(\epsilon\gamma/4) + |D| 
\leq 6/(\epsilon\gamma)$. 

Note that $|x_{t+1}| \leq \epsilon\gamma/4$. Taking $\alpha = 2$ in \cite[Lemma 3.3]{GLPS17}, we know that
\[
\|(\Phi(x-x_{[t]}))_{\Gamma(D)}\|_1 \leq 4\cdot\beta\epsilon d\left(\frac{3}{2} 
+ 2s\cdot\frac{\epsilon\gamma}{4}\right) \leq 8\beta\epsilon d.
\]

By the isolation property, there are at most $\frac{6}{\epsilon\gamma}\cdot \frac{\epsilon\theta}{12} = \frac{\theta}{2\gamma}$ elements in $T$ which are not isolated in at least $(1-\zeta) d$ nodes from other elements in $T$. This implies that at least $\theta/(2\gamma)$ elements in $D$ 
are isolated in at least $(1-\zeta) d$ nodes from other elements in $T$.

A decoy at position $i$ receives at least $\epsilon\gamma/4$ noise in at least $(\beta-\zeta) d$ isolated nodes of $\Gamma(\{i\})$, hence in total, a decoy element receives at least $\epsilon\gamma(\beta-\zeta)d/4$ noise. Therefore at least $\theta/(2\gamma)$ decoys overall should receive noise at least
\[
\frac{\epsilon\gamma(\beta-\zeta)d}{4}\cdot \frac{\theta}{2\gamma} > 8\beta\epsilon d\geq \|(\Phi(x-x_{[t]}))_{\Gamma(D)}\|_1,
\]
which is a contradiction. Therefore there are at most $\theta s$ decoys. 
\end{proof}

We remark that the construction in~\cite{GLPS17} is similar to the partition setup in~\cite{LNNT}, where the linking information and the message block are `absorbed' into the fashion coordinates are split into buckets so that recovering a heavy hitter in a bucket will automatically recover that information correctly instead of recovering the information from second-layer buckets with an error correcting code. In this paper, however, we opt for the two-layer construction for the `for-all' guarantee, for the presentation would be simpler for our sparse recovery results as some auxiliary lemmata are already proved in~\cite{GLPS17}.

\section{A Sublinear Time $\ell_{1}/\ell_1$ Algorithm} \label{sec:intro}

The result is almost immediate by replacing the list-recoverable code in~\cite{GLPS17} with the clustering algorithm in~\cite{LNNT}, which we now give a brief review. The overall algorithm is an iterative algorithm of $\Theta(\log k)$ iterations. Each iteration is called a \textit{weak system}, which (i) recovers at least a constant fraction of the remaining heavy hitters (and hence all heavy hitters can be recovered in $\Theta(\log k)$ iterations) and (ii) introduces only a small amount of error (see, e.g.~\cite{PS12,GLPS17}). The introduced error comes from two sources: the estimation error of those recovered heavy hitters and some small coordinates that can be safely ignored. 

We now present the precise statement of our weak system 
, which is central to our $\ell_1/\ell_1$ result. The proof is almost identical to that in~\cite{GLPS17}, nevertheless we include it for completeness.

\begin{lemma}[Weak system]\label{lem:weak_ell1/ell_1}
Suppose that $s\leq \sqrt{n}$ and $\epsilon\in(0,1)$. There exist a linear sketch $\Phi\in \R^{m\times n}$ and an algorithm $\textsc{WeakSystem}(x,s,\epsilon)$ satisfying the following:
\begin{itemize}[parsep=0pt,partopsep=0pt]
	\item For any vector $x\in \R^n$ that can be written as $x = y+z$, where $y$ and $z$ have disjoint supports, $|\supp(y)|\leq {s}$, $\|y\|_\infty\geq \epsilon/(2s)$ and $\|z\|_1\leq 3/2$, given the measurements $\Phi x$, the decoding algorithm $\mathcal{D}$ returns $\hat x$ such that $x$ admits the decomposition of 
	\begingroup
\setlength{\abovedisplayskip}{3pt}
\setlength{\belowdisplayskip}{3pt}	
	\[x = \hat x + \hat y + \hat z,\]
	\endgroup
where $|\supp(\hat x)| = s$, $|\supp(\hat y)|\leq s/8$ and $\|\hat z\|_1\leq \|z\|_1+\epsilon/4$. Intuitively, $\hat y$ and $\hat z$ will be the head and the tail of the residual $x-\hat x$, respectively;
	\item $m = \Oh(\epsilon^{-2}s\log n)$;
	\item $\mathcal{D}$ runs in $\Oh(s^{3}\poly(1/\epsilon,\log n))$ time.
\end{itemize}
\end{lemma}
\begin{proof}
We follow the construction and the argument as in~\cite{GLPS17}. We instantiate the two-layer hashing and the encoding scheme as in \Cref{sec:skeleton}, where $\alpha\in (1,2)$, $B_1 = \Theta(s^\alpha/\epsilon^{2\alpha})$, $d_1 = \Theta(\epsilon^{-1}\frac{\log n}{\log(B_1/s)})$, $B_2 = \Theta(s/\epsilon)$ and $d_2 = \Theta(\log(B_1/s))$. By \Cref{lem:two-layer-isolation} we can find a two-layer hashing with these prescribed parameters which satisfies  $(4s,d_1d_2,O(\epsilon))$-expansion property and $(O(s/\epsilon),O(\epsilon),\Theta(1))$-isolation property. It is also easy to verify that the length of each message block $\bar m_{i,r}$ is $L = \Theta(\log(B_1/s)) + \Delta\log(B_1) \leq d_2/2$ if we choose $d_2$ large enough. We can use two second-layer measurements to encode $1$ bit of message by replacing a single entry $a$ in the measurement matrix with a $2\times 1$ block of $\left(\begin{smallmatrix} a \\ 0\end{smallmatrix}\right)$ or $\left(\begin{smallmatrix} 0 \\  a\end{smallmatrix}\right)$, depending on the bit to encode. To decode the bit, suppose that the two corresponding measurements are $\left(\begin{smallmatrix}a \\ b\end{smallmatrix}\right)$ and we convert them back to $0$ if $|a|<|b|$ and $1$ if $|a|\geq |b|$. When the heavy hitter is isolated and the noise is small in a bucket, the bit is expected to be the corresponding bit in the message for that heavy hitter. See Section~\ref{sec:example1} for a toy example of such scheme. 

Next we shall show that we can recover most bits of the message for at least a large constant fraction of heavy hitters. Invoking Lemma~\ref{lem:decoy} with $\delta=O(1)$, $\theta=O(1)$, $\beta=O(1)$ and $\gamma=1/s$, and following the argument in~\cite[Section 4.1]{GLPS17}, we have good estimates for all but at most $s/8$ heavy hitters (elements in $\supp(y)$). Call those heavy hitters \emph{well-estimated}. The two-layer hashing eventually hashes $n$ coordinates into $B_2$ buckets and repeat $d_1 d_2$ times, and we know that each well-estimated heavy hitter $i$ receives small noise in at least $(1-\delta)d_1 d_2$ repetitions. This implies that there exist $\delta_1$ and $\delta_2$ such that for each well-estimated heavy hitter $i$, there exist $(1-\delta_1)d_1$ first-layer repetitions such that in each such first-layer repetition $r$ the heavy hitter $i$ receives small noise in at least $(1-\delta_2)d_2$ second-layer repetitions. For each such pair $(i,r)$, we can recover at least $(1-\delta_2)$ fraction of the message $\bar m_{i,r}$, and if we protect $\bar m_{i,r}$ using a constant-rate error-correcting code (e.g.~Reed-Solomon code) that can tolerate up to $\delta_2$ fraction of error, we shall recover $\bar m_{i,r}$ in entirety. To summarize, for each well-estimated heavy hitter $i$, we can recover $\bar m_{i,r}$ for at least $(1-\delta_1)d_1$ values of $r \in [d_1]$. We note that $\delta_1$ can made arbitrarily small by adjusting the constants in the two-layer construction and making $\delta$ arbitrarily small.

Now we construct the chunk graph as in~\cite{LNNT}. 
The chunk graph has $B_1 d_1$ nodes, indexed by pairs $(b,r)$ for $b\in[B_1]$ and $r\in [d_1]$. For each bucket $b$ in the first-repetition $r$, we recover a message of length $L$, break it up into blocks of the same structure as in $\hat m$ and extract the linking information $q_1(b,r),\dots,q_\Delta(b,r)$. We say in $\tilde G$ the node $(b,r)$ makes suggestion to connect to $(q_\ell(b,r),\Gamma_\ell(r))$, and we add an edge if both endpoints suggest each other. 
By the argument in \cite[Lemma 2]{LNNT}, a well-estimated heavy hitter $i$ corresponds to an $\epsilon_0$-spectral cluster of $\tilde G$ for some small $\epsilon_0 > 0$. The spectral clustering algorithm~(\cite[Theorem 1]{LNNT}) will find all those spectral clusters, recovering a constant fraction of message $m_i$ and enabling us to identify the index $i$ of the heavy hitter. In each first-layer we retain the bucket of magnitude at least $\epsilon/(4s)$ so there are $O(s/\epsilon)$ buckets, and we therefore have a candidate list of size $\Oh(s/\epsilon)$ which misses at most $s/8$ heavy hitters. Finally we evaluate every candidate and retain the biggest $s$ ones (in magnitude).

Each recovered coordinate is estimated to within $\epsilon/(4s)$, thus
\[
\|\hat z\|_1 \leq \|z\|_1 + |\supp(\hat x)|\cdot \frac{\epsilon}{4s} \leq \|z\|_1 + \frac{\epsilon}{4}.
\]

The total number of measurements is $B_2 d_1 d_2 = \Oh(\epsilon^{-2} s\log n)$. For each first-layer repetition, we enumerate all coordinates in the bucket of size $B_1$, and decode the associated message (which is of length $d_2$), which takes time $\Oh(B_1 \poly(d_2)) = \Oh(s^{\alpha}\poly(1/\epsilon,\log n))$. We then run the spectral clustering algorithm on a graph of size $O(s/\epsilon\cdot d_1)$ in time $\tilde\Oh((s/\epsilon\cdot d_1)^3) = \Oh(s^3\poly(1/\epsilon, \log n))$. To obtain the indices of the candidates, We decode $\Oh(s/\epsilon\cdot d_1) = \Oh(\epsilon^{-2} s \log n)$ messages, and the decoding algorithm on each $\Theta(\log n)$-bit-long $m_i$ with a constant fraction corruption runs in time $O(\poly(\log n))$. Lastly we estimate each of the candidates and retain the biggest $s$ ones, which takes time $\Oh((s/\epsilon)d_1\cdot B_2d_1d_2) = \Oh(s^2 \poly(1/\epsilon,\log n))$. The overall runtime is dominated by the clustering algorithm and is therefore $\Oh(s^3\poly(1/\epsilon, \log n))$.
\end{proof}
\section{A Sublinear Time $\ell_{\infty}/\ell_1$ Algorithm} \label{sec:ell_infty}


For ease of exposition and connection with the previous algorithms, we set $k= \lceil 1/\epsilon \rceil$ in this section.

First, we prove the following weaker theorem, and shall show how to bootstrap this theorem in order to obtain Theorem \ref{thm:ell_infty/ell_1} in \Cref{sec:combining}.

\begin{theorem}\label{thm:weakerell_infty/ell_1}
There exists a linear sketch $\Phi \in \mathbb{R}^{m \times n}$ such that for every $x\in\R^n$, we can, given $\Phi x$, find an $\Oh(k)$-sparse vector $\hat{x}$ such that
$\|x-\hat{x}\|_{\infty} \leq \frac{1}{k} \|x_{-k}\|_1$ 
in $\Oh(k^6\poly(\log n))$ time. The number of rows of $\Phi$ equals $m = \Oh(k^2\log n\log^\ast k)$. The space needed to store $(y,\Phi)$ is $\Oh(k^2 \log n\cdot \log^\ast k\cdot \log \log k)$ words.
\end{theorem}

We remark that this theorem is slightly weaker than our main result, namely Theorem \ref{thm:ell_infty/ell_1}, because the error is measured with respect to $\|x_{-k}\|_1$, and not with respect to $\|x_{-k^2}\|_1$. We are going to bootstrap Theorem \ref{thm:weakerell_infty/ell_1} later.

\paragraph{Weak-Level System }
\Cref{lem:decoy} is a central argument for deterministic sparse recovery tasks. Previous works~\cite{PS12,GLPS17} used the lemma with $\gamma = 1/s$ and constant $\theta \in (0,1)$ to show that if we estimate every coordinate $x_i$ to be the median of $E_i$ and take the biggest $\Theta(s)$ estimates in magnitude, we shall miss at most $2 s$ heavy hitters, upon which weak systems that miss a $\theta$-fraction of heavy hitters were constructed. The overall algorithm makes sequential calls to weak systems with geometrically decreasing number of remaining heavy hitters. In our case, since we want the stronger $\ell_{\infty}/\ell_1$ guarantee, we are not allowed to decrease geometrically the number of rows for the weak systems. But, with more allotted number of rows, we can recover much more than a constant fraction of heavy hitters by exploiting the power of $\theta$.

\begin{lemma}[Weak system]\label{lem:weak}
Suppose that $w\leq s\leq k\leq \sqrt{n}$ and $\eta \in (0,1)$ is an arbitrarily small constant. There exist a linear sketch $\Phi\in \R^{m\times n}$ and an algorithm $\textsc{WeakSystem}(x,k,s,w)$ satisfying the following:
\begin{itemize}[parsep=0pt,partopsep=0pt]
	\item For any vector $x\in \R^n$ that can be written as $x = y+z$, where $y$ and $z$ have disjoint supports, $|\supp(y)|\leq {s}$, $\|y\|_\infty\geq 1/(2k)$ and $\|z\|_1\leq 2-s/k$, given the sketch $\Phi x$, the decoding algorithm $\mathcal{D}$ returns $\hat x$ such that $x$ admits the decomposition of 
	\begingroup
\setlength{\abovedisplayskip}{2pt}
\setlength{\belowdisplayskip}{2pt}	
	\[x = \hat x + \hat y + \hat z,\]
	\endgroup
	where $|\supp(\hat x)| = s$, $\|(x-\hat x)_{\supp(\hat x)}\|_\infty\leq \frac{1}{2k}$, $|\supp(\hat y)|\leq \sqrt{sw}$ and $\|\hat z\|_1\leq 2- \frac{s}{2k}$. Intuitively, $\hat y$ and $\hat z$ will be the head and the tail of the residual $x-\hat x$, respectively.
	\item $m = \Oh(\frac{k^2}{w}\log n)$, 
	\item $\mathcal{D}$ runs in $\Oh(k^6 \poly(\log n))$ time. 
\end{itemize}
\end{lemma}
\begin{proof}
We follow the argument in~\cite{GLPS17}. We instantiate the two-layer hashing and the encoding scheme as in \Cref{sec:skeleton}, where $\alpha\in (1,2)$, $B_1 = \Theta(k^{2\alpha})$, $d_1 = \Theta(\frac{k}{\sqrt{sw}}\frac{\log n}{\log(B_1/s)})$, $B_2 = \Theta(k\sqrt{s/w})$ and $d_2 = \Theta(\log(B_1/s))$. By \Cref{lem:two-layer-isolation} (to see the conditions hold, replace $k$ with $s$ and $\epsilon$ with $k/\sqrt{sw}$), we can find a two-layer hashing with these prescribed parameters which satisfies  $(\Theta(k),d_1d_2,\frac{\sqrt{sw}}{k})$-expansion property and $(\Theta(k),\frac{\sqrt{sw}}{k},\Theta(1))$-isolation property. The constants in the $\Theta$-notations above all depend on $\eta$. It is also easy to verify that the length of each message block $\bar m_{i,r}$ is $L = \Theta(\log(B_1/s)) + \Delta\log(B_1) \leq d_2$ if we choose $d_2$ large enough.

Invoking \Cref{lem:decoy} with $\delta = \Theta(\sqrt{sw}/k)$, $\theta = \Theta(\sqrt{sw}/k)$, $\epsilon = 1$ and $\gamma = 1/k$, and following the argument in~\cite[Section 4.1]{GLPS17}, we have good estimates for all but at most $\Theta(\theta/\gamma) = \sqrt{sw}$ heavy hitters (elements in $\supp(y)$). Call those heavy hitters \emph{well-estimated}. 
Following the same argument as in the proof of \Cref{lem:weak}, we can, for each well-estimated heavy hitter $i$, recover $\bar m_{i,r}$ for at least $(1-\delta_1)d_1$ values of $r \in [d_1]$. We note that $\delta_1$ can made arbitrarily small by adjusting the constants in the two-layer construction and making $\delta$ arbitrarily small.

We construct the chunk graph $\tilde G$ as in~\cite{LNNT}. 
By the argument in \cite[Lemma 2]{LNNT}, a well-estimated heavy hitter $i$ corresponds to an $\epsilon_0$-spectral cluster of $\tilde G$ for some small $\epsilon_0 > 0$. The spectral clustering algorithm~(\cite[Theorem 1]{LNNT}) will find all those spectral clusters, recovering a constant fraction of message $m_i$ and enabling us to identify the index $i$ of the heavy hitter. In each first-layer we retain the bucket of magnitude at least $1/(4k)$ so there are $O(k)$ buckets, and we therefore have a candidate list of size $\Oh(k)$ which misses at most $\sqrt{sw}$ heavy hitters. Finally we evaluate every candidate and retain the biggest $s$ ones (in magnitude).

Each recovered coordinate is estimated to within $\epsilon \gamma/4 \leq 1/(4k)$, thus
\[
\|\hat z\|_1 \leq \|z\|_1 + \frac{|\supp(\hat x)|}{4k} \leq 2 - \frac{s}{k} + \frac{s}{2k} \leq 2 - \frac{s}{2k}.
\]

The total number of rows is $B_2 d_1 d_2 = \Oh(\frac{k^2}{w}\log n)$. For each first-layer repetition, we enumerate all coordinates in the bucket of size $B_1$, and decode the associated message (which is of length $d_2$), which takes time $\Oh(B_1 \poly(d_2)) = \Oh(k^{2\alpha}\poly(\log k))$. We then run the spectral clustering algorithm on a graph of size $\Oh(k d_1)$ in time $\Oh((k d_1)^3) = \Oh(k^6\poly(\log n))$. To obtain the indices of the candidates, We decode $\Oh(k d_1) = \Oh(k^2\log n)$ messages, and the decoding algorithm on each $\Theta(\log n)$-bit-long $m_i$ with a constant fraction corruption runs in time $O(\poly(\log n))$. Lastly we estimate each of the candidates and retain the biggest $s$ ones, which takes time $\Oh(kd_1\cdot B_2d_1d_2) = \Oh(k^4\log^2 n)$. The overall runtime is dominated by the clustering algorithm and is therefore $\Oh(s^6\poly(\log n)) = \Oh(k^6\poly(\log n))$.
\end{proof}

\begin{algorithm}[t]
\caption{Overall algorithm for $\ell_\infty/\ell_1$ sparse recovery. {In the pseudocode below, $v^{(i,r)}$ as an argument of \textsc{WeakSystem} is understood to be restricted to the corresponding coordinates.}}\label{alg:ellinfell1}
\begin{algorithmic}[1]
	\Require sketching matrix $\Phi$, sketch $v = \Phi x$, sparsity parameter $k$
	\Ensure $\hat x$ that approximates $x$ with $\ell_\infty/\ell_1$ guarantee
	\State $v^{(0,1)} \leftarrow v$
	\State $k_1 \gets k$
	\While{$k_r > 4$}\Comment{$\Oh(\log^\ast k)$ rounds}
		\State $i\gets 0$
		\While{$k_r^{2^{-i}}\geq \max\{(i+1)^2, 4(1+\frac{1}{i+1})^4\}$}\Comment{$\Oh(\log\log k_r)$ steps}
			\State $s_{i,r} \leftarrow (i+1)^2 k_r^{2^{-i}}$ 
			\State $w_{i,r} \leftarrow (i+1)^2$
			\State $\hat{x}^{(i,r)} \leftarrow \textsc{WeakSystem}(\Phi^{(i,r)},v^{(i,r)},k,s_{i,r},w_{i,r})$
			\State $v^{(i+1,r)} \leftarrow v^{(i,r)} - \Phi\hat{x}^{(i,r)}$
			\State $i\gets i+1$
		\EndWhile
		\State $i_r^\ast \gets i-1$
		\State $s_{i,r} \leftarrow (s_{i_r^\ast,r})^{2^{-(i- i_r^\ast-1)}}$
		\While{$s_{i,r}\geq \max\{4,\log\log k_r\}$}\Comment{$\Oh(1)$ steps}
			\State $w_{i,r} \leftarrow 1$
			\State $\hat{x}^{(i,r)} \leftarrow \textsc{WeakSystem}(\Phi^{(i,r)},v^{(i,r)},k,s_{i,r},w_{i,r}) $
			\State $v^{(i+1,r)} \leftarrow v^{(i,r)} - \Phi\hat{x}^{(i,r)}$
			\State $i\gets i+1$
			\State $s_{i,r} \leftarrow (s_{i_r^\ast,r})^{2^{-(i- i_r^\ast-1)}}$
		\EndWhile
		\State $i_r^+\gets i - i_r^\ast - 1$
		\State $v^{(0,r+1)} \leftarrow v^{(i,r)}$
		\State $k_{r+1} \leftarrow s_{i,r}$
		\State $r\gets r+1$
	\EndWhile
	\State $\hat{x}_{\text{final}} \leftarrow \textsc{WeakSystem}(\Phi^{(0,r)},v^{(0,r)},k,4,1/5)$\Comment{last round}
	\State \Return{$\hat{x} \gets \hat{x}_{\text{final}} + \sum_{r,i} \hat x^{(i,r)}$}
\end{algorithmic}
\end{algorithm}

\paragraph{Construction of Measurement Matrix }
Now we construct the sketch for Theorem \ref{thm:weakerell_infty/ell_1}. The main idea is to apply the weak system (\Cref{lem:weak}) repeatedly. We form our sketching matrix $\Phi$ as illustrated below and present our recovery algorithm in \Cref{alg:ellinfell1}. 
\[
\Phi = \left[\begin{array}{c}
			\Phi_1\\ 
			\Phi_2\\ 
			\vdots\\
			\Phi_R\\
			\Phi_{\text{final}}
		\end{array}
	  \right]	,
	  \qquad \text{where }
\Phi_r = \left[\begin{array}{c}
			\Phi^{(1,r)}\\
			\vdots\\ 
			\Phi^{(i_r^\ast,r)}\\ 
			\Phi^{(i_r^\ast+1,r)}\\
			\vdots\\ 
			\Phi^{(i_r^\ast+i_r^+,r)}
		\end{array}
	  \right], r=1,\dots,R.
\]
Here 
\begin{itemize}[noitemsep,topsep=0pt,parsep=0pt,partopsep=0pt]
\item the overall $\Phi$ is the vertical concatenation of $R+1$ matrices and every layer, except the last one, is further a concatenation of $i_r^\ast + {i_r^+}$ matrices, where $R = \Theta(\log^\ast k)$ and $i_r^\ast,i_r^+$ are computed as in  \Cref{alg:ellinfell1};
\item the $i$-th layer in $\Phi_r$, namely $\Phi_{i,r}$, is the sketching matrix for the weak system (\Cref{lem:weak}) with parameters $s = s_{i,r}$ and $w = w_{i,r}$, the values of which are as assigned in \Cref{alg:ellinfell1};
\item the last layer of $\Phi$, namely $\Phi_{\text{final}}$, is the sketching matrix for the weak system with parameters $s = 4$ and $w = 1/5$.
\end{itemize}

Overall there are $i_r^\ast + i_r^+ + 1$ iterations in the algorithm and each iteration corresponds to one block of $\Phi$. There are $k$ heavy hitters at the beginning. Each iteration reduces the number of remaining heavy hitters to almost its square root and hence in $O(\log \log k)$ iterations the number of remaining heavy hitters will be reduced to a constant, and those heavy hitters will be all recovered in the last iteration. In order to minimize the number of measurements, in each of the first $i_r^\ast$ iterations, the number of remaining heavy hitters is reduced to slightly bigger than its square root, and in each of the next $i_r^+$ iterations, to exactly the square root.

The parameters $s_{i,r}, w_{i,r}, i_r^\ast, i_r^+$ may seem adaptive at the first glance, but they in fact do not depend on the input $x$ and depend only on the sparsity parameter $k$ and can thus be pre-computed. The whole algorithm is non-adaptive.

\paragraph{Proof of Theorem \ref{thm:weakerell_infty/ell_1}} We only provide a sketch of the proof below and leave the full proof to \Cref{sec:weaker_ellinfty/ell1_proof}.

\begin{proof}[Proof Sketch]
Without loss of generality, assume that $\|x_{-k}\|_1 = 1$. We shall apply \Cref{lem:weak} repeatedly to obtain a sequence of vectors $\hat{x}^{(i,r)}$, which admit decompositions $x^{(i,r)} = x - \hat{y}^{(i,r)} - \hat{z}^{(i,r)}$. We can show inductively that the loops invariants (I) below, parametrized by $(i,r)$, are satisfied at the beginning on each while loop from Line 5 to 11 in Algorithm~\ref{alg:ellinfell1} and the loop invariants (II) below are satisfied at the beginning on each while loop from Line 14 to 20.
\[
\text{(I)}\left\{\begin{aligned}
|\supp(\hat{y}^{(i,r)})| &\leq s_{i,r} := (i+1)^2 k_r^{2^{-i}},\\
\|\hat{z}^{(i,r)}\|_1 &\leq 2- s_{i,r}/k;
\end{aligned}\right.
\qquad
\text{(II)}\left\{\begin{aligned}
|\supp(\hat{y}^{(i,r)})| &\leq s_{i,r} := (s_{i_r^\ast,r})^{2^{-(i-i_r^\ast-1)}},\\
\|\hat{z}^{(i,r)}\|_1 &\leq 2 - s_{i,r}/k.
\end{aligned}\right.
\]

When the algorithm runs into Line 25, that is, when there are at most $4$ heavy hitter left, we shall recover all of them in one call to the weak system. 

The total number of rows and runtime, etc., follow from direct calculations.
\end{proof}

\subsection{Getting the Final Result}\label{sec:combining}

We now show how to combine the $\ell_1/\ell_1$ scheme with the $\ell_{\infty}/\ell_1$ scheme to obtain the main result of the paper. For completeness, we restate the main theorem with the substitution of $k=\lceil 1/\epsilon\rceil$.

\begin{reptheorem}{thm:ell_infty/ell_1}[rephrased]
There exists a linear sketch $\Phi \in \mathbb{R}^{m \times n}$ such that for every $x\in\R^n$, we can, given $\Phi x$, find an $\Oh(k)$-sparse vector $\hat{x}$ such that
$	\|x-\hat{x}\|_{\infty} \leq (1/k)\|x_{-k^2}\|_1$ 
in $\Oh(k^{6}\poly(\log n))$ time. The sketch length is $m = \Oh(k^2 \log n\log^\ast k)$ and the space needed to store $(y,\Phi)$ is $\Oh(k^2 \log n \cdot \log^\ast k\cdot \log\log k)$ words. 
\end{reptheorem}

We shall need the following lemma from \cite{nnw12}.

\begin{lemma}[Point Query \cite{nnw12}]\label{lem:incoherent} 
There exists a matrix $C \in \mathbb{R}^{m \times n}$ with $m = \Oh(k^2 \log n)$ rows, such that given $y = C x$ and $i \in [n]$, it is possible to find in $\Oh(k\log n)$ time a value $\hat{x}_i$ such that
$	|x_i - \hat{x}_i| \leq (1/k) \|x_{[n]\setminus\{i\}}\|_1$.
\end{lemma}

The construction of $C$ given in~\cite{nnw12} includes taking $C$ to be a Johnson-Lindenstrauss Transform matrix for the set of points $\{0,e_1,\dots,e_n\}$, where $e_1,\dots,e_n$ is the canonical basis of $\R^n$. 

We are now ready to prove \Cref{thm:ell_infty/ell_1}.

\begin{proof}
We first pick a matrix $A$ using \Cref{thm:ell_1/ell_1}, setting the sparsity parameter to $k^2$ and $\epsilon=1$. We also pick a matrix $B$ satisfying the guarantees of \Cref{thm:weakerell_infty/ell_1}, with sparsity $6k$, and a matrix $C$ using \Cref{lem:incoherent} with sparsity parameter $6k$. Our sketching matrix $\Phi$ is the vertical concatenation of $A$, $B$ and $C$. The total number of rows is $\Oh(k^2 \log n)$ for $A$ and $C$, and $\Oh(k^2 \log n \log^{\ast} k)$ for $B$, for a total of $\Oh(k^2 \log n \log^\ast k)$ rows.

We first run the algorithm on $Ax$ to obtain an $\Oh(k^2)$-sparse vector $z$ such that $ \|x-z\|_1 \leq 2 \|x_{-k^2}\|_1$. Then we form $B(x-z)$ and using the query algorithm for $B$, we find an $\Oh(k)$-sparse vector $w$ such that
\begin{equation}\label{eqn:x-z-w}
	\|(x-z) - w\|_{\infty} \leq \frac{1}{2k} \|x-z\|_1 \leq \frac{1}{k} \|x_{-k^2}\|_1.	
\end{equation}

Let $\mathcal{H}$ be the set of coordinates $i \in [n]$ such that $|x_{i}| > \frac{1}{k} \|x_{-k^2}\|_1$. We claim that $\mathcal{H} \subseteq \supp(z) \cup \supp(w)$; otherwise, it holds for $i\in \mathcal{H}$ that
\[ 
|((x-z) - w)_i| = |x_i| > \frac{1}{k}\|x_{-k^2}\|_1, 
\] 
which contradicts \eqref{eqn:x-z-w}. The next step is to estimate $x_i$, for every $i \in \supp(z) \cup \supp(w)$, up to $(1/k)\|x_{-k^2}\|_1$ error. This argument is almost identical to \cite{nnw12}, but we include it here for completeness. For every such $i$, define vector $z'$ to be equal to $z$ but with the $i$-th coordinate zeroed out. Then we run the point query algorithm of Lemma \ref{lem:incoherent} on sparsity parameter $6k$ with sketch $C(x-z')$ to obtain a value $\hat{x}_i$ such that
\[	
|\hat{x}_i - x_i| = |\hat{x}_i - (x-z')_i| \leq \frac{1}{6k} \|(x-z')_{[n] \setminus \{i\}} \|_1 \leq \frac{1}{6k}\|x-z\|_1 \leq \frac{1}{3k}\|x_{-k^2}\|_1.	
\]
We note that $|\mathcal{H}|\leq k$ and, hence, by keeping the top $4k$ coordinates in magnitude, we shall include all elements in $\mathcal{H}$. Otherwise, there are at least $3k$ estimates of value at least $\frac{2}{3k}\|x_{-k^2}\|_1$ and so there are at least $3k$ coordinates of $x_{\supp(z) \cup \supp(w)}$ of magnitude at least $\frac{1}{3k}\|x_{-k^2}\|_1$, which is impossible. This concludes the proof of correctness.

\paragraph{Running time.} The first step of obtaining $z$ takes time $\Oh(k^3\poly(\log n))$ by \Cref{thm:ell_1/ell_1}. The second step of obtaining $w$ takes time $\Oh(k^6\poly(\log n))$ by \Cref{thm:weakerell_infty/ell_1}. The third step makes $\Oh(k^2)$ point queries. For each point query, it computes $C(x-z') = Cx - Cz'$, where $Cx$ is part of the overall sketch and $Cz'$ can be efficiently computed in $\Oh(k^4\log n)$ time since $C$ has $\Oh(k^2\log n)$ rows and $z'$ is $\Oh(k^2)$-sparse. Then the point query procedure itself runs in time $\Oh(k\log n)$ by \Cref{lem:incoherent}. The total runtime of the third step is thus $\Oh(k^6\log n)$. The overall runtime is dominated by that of the second step.

\paragraph{Storage space.} The space to store $A$ is $\Oh(k^2\log n)$ words by \Cref{thm:ell_1/ell_1}. The space to store $B$ is $\Oh(k^2 \log n \cdot \log^\ast k\cdot \log\log k)$ words by \Cref{thm:weakerell_infty/ell_1}. The space to store $C$ is $\Oh(k^2\log n)$ words by taking $C$ to be a fast Johnson-Lindenstrauss Transform matrix~\cite{KN14}. The overall storage space is dominated by that of $B$.
\end{proof}
\section{Strict Turnstile Model} \label{sec:strictturnstile}

In this section we give constructions of strongly explicit matrices that allow sublinear decoding in the strict turnstile model. In Section~\ref{sec:pointquery}, we show that the explicit incoherent family of matrices of \cite{nnw12} gives also the tail guarantee. In Section~\ref{sec:sublinear_explicit}, we show have to recursively combine those incoherent matrices to obtain sublinear decoding time.

\subsection{Point Query}\label{sec:pointquery}

The following theorem appears in \cite{nnw12}. The construction of the sketching matrix is based on Reed-Solomon codes. In general, given a code of $\mathcal{C} = \{C_1, \ldots, C_n\}$ of alphabet size $q$ and block length $b$, we create a $qb \times n$ matrix, where the $i \in [n]$  column has a $1$ in position $\alpha \cdot (q-1) + C_i(\alpha), \forall \alpha \in [b]$, and $0$ otherwise. The following theorem is obtained by instantiating the construction above with a Reed-Solomon code of an appropriate alphabet size and block length.

\begin{lemma}[\cite{nnw12}] \label{lem:explicit}
There exists a strongly explicit matrix $\Phi \in \mathbb{R}^{m\times n}$, with  
\[
m=\Oh\left(k^2 \left(\frac{\log n}{ \log \log n + \log k} \right)^2\right),	
\]
such that 
given $v = \Phi x$ and $i \in [n]$, we can find a value $\hat{x}_i$ such that $|x_i - \hat{x}_i| \leq \frac{1}{2k}\|x_{-1}\|_1$ in time $\Oh(k \log n/(\log \log n+ \log k))$.
\end{lemma}

The main  downside of this lemma is that gives the weaker tail guarantee with $r=1$, and not with $r=k$. We show that it is possible to use the exact same matrix as in \cite{nnw12} to obtain the tail guarantee with $r=k$. For that result, we exploit the assumption that the word size is $w = \Theta(\log n) $. Our result is the following.

\begin{theorem} \label{thm:tail_point_query}
There exists a strongly explicit matrix $\Phi \in \mathbb{R}^{m\times n}$ with  
\[
m=\Oh\left(k^2 \left(\frac{\log n}{ \log \log n + \log k} \right)^2\right)
\] 
and an algorithm which,
given $v= \Phi x$ and a subset $S \subseteq [n]$ such that $H(x,k,1) \subseteq S$, 
returns an $O(k)$-sparse vector $\hat{x} \in \mathbb{R}^n $ such that 
\[
\|x - \hat{x}\|_{\infty} \leq \frac{1}{k}\|x_{-k}\|_1
\]
in time $\Oh(|S|k \log^2 n/(\log \log n+ \log k))$.
\end{theorem}

We choose $\beta, c > 1$ sufficiently large constants, and a matrix $\Phi$ from Lemma \ref{lem:explicit} with $k \gets \beta k$. Our query algorithm is presented in Algorithm \ref{alg:reed_solomon_point_query}. The following is a lemma crucial for the correctness of our algorithm.

\begin{algorithm}
\caption{Strongly explicit $\ell_\infty/\ell_1$ sparse recovery}\label{alg:reed_solomon_point_query}

\begin{algorithmic}[1]
	\Require sketching matrix $\Phi$, sketch $v = \Phi x$, sparsity parameter $k$, set $S$
	\Ensure $\hat x$ that approximates $x$ with $\ell_\infty/\ell_1$ guarantee
	\State $ v^{(0)} \leftarrow v$
	\State $ w^{(0)} \leftarrow 0$
	\State $R \leftarrow c \log n$
	\For {$t=0$ to $R$}
		\State $\hat{x}^{(t)} \leftarrow \textsc{ReduceNoise}(v^{(t)},S)$
		\State $ w^{(t+1)} \leftarrow w^{(t)} + \hat{x}^{(t)}$
		\State $ v^{(t+1)} \leftarrow v^{(t)} - \Phi \hat{x}^{(t)}$
	\EndFor
	\State \Return $ w^{(R+1)}$

	\Statex
\Procedure{ReduceNoise}{$ u,S$} \Comment{$u= \Phi z$}
		\For {$i \in S $}
			\State $\hat{z}_i \leftarrow \mathrm{median}_{q: \Phi_{q,i} \neq 0}\ u_q$
		\EndFor
		\State $T\gets$ the set of the $5k$ largest coordinates $\hat z_i$ in magnitude.
		\State \Return $\hat{z}_T$ 
\EndProcedure
\end{algorithmic}
\end{algorithm}

\begin{lemma}\label{lem:one_invocation}

Let $z  \in \mathbb{R}^n$. If $ \|z_{S}\|_1 \geq 3\|z_{\bar{S}}\|_1$ 
then one invocation of \textsc{ReduceNoise}($\Phi z, S$) yields a vector $\hat{z}$ such that $z' = z - \hat{z}$ satisfies
		\[	
			\|z'\|_1  \leq \gamma \|z\|_1,	
		\]
for some absolute constant $\gamma \in (0,1)$. 
\end{lemma}

\begin{proof}

We note that in what follows, $T$ is the set obtained in line 14, in procedure \textsc{ReduceNoise}.

Define two sets
\[ S_1 = \left\{ i \in S: |z_i| \geq \frac{1}{3k}\|z_{S}\|_1\right\}, \]
and 
\[	S_2 = \left\{i \in S: |z_i| \geq \frac{1}{4k} \|z_{S}\|_1\right\}.	\]


We also have the assumption that the set $S$ contains all the $1/k$-heavy hitters of $z$, and hence $S_1 \subseteq S$. Furthermore, the estimates $\{\hat{z}_i\}_{i \in S}$ satisfy
\[ 
|\hat{z}_i - z_i | \leq \frac{1}{\beta k} \|z\|_1	
\]
by Lemma \ref{lem:explicit}. Moreover, $|S_2| \leq 4k$. This immediatelly yields $S_1 \subseteq T$, because no $i \in S_1$ can be displaced by some $i' \notin S_2$, i.e.
	\[	\forall i \in S_1,i' \in S\setminus S_2: |\hat{z}_i| > |\hat{z}_{i'}|.	\]
Thus, due to line 7 of Algorithm \ref{alg:reed_solomon_point_query} we have that 
\begin{align*}
	\|z_{S}'\|_1 &= \|z_{S_1}'\|_1 + \|z_{S \setminus S_1}'\|_1 \\
&\leq \|z_{S_1}'\|_1 + \left(\|z_{S \setminus S_1} \|_1 + k \cdot \frac{1}{\beta k} \|z\|_1 \right)\\
&\leq k \cdot \frac{1}{\beta k} \|z\|_1 + k\frac{1}{3k} \|z_{S}\|_1 + \frac{1}{\beta}\|z\|_1\\
&=\frac{2}{\beta}\|z\|_1 + \frac{1}{3} \|z_{S}\|_1 \\
&= \left(\frac{2}{\beta} + \frac{1}{3}\right)\|z_{S}\|_1 + \frac{2}{\beta}\|z_{\bar{S}}\|_1,
\end{align*}
where the first inequality follows from the fact that $S \subseteq T$ and the correctness of the estimates, and the second inequality follows from the correctness of the estimates and the fact that $|z_i|\leq 1/(3k)\|z_{S}\|_1$ for every $i \in S\setminus S_1$.	

Overall we have that
\[	
\|z'\|_1 = \|z_{S}'\|_1 + \|z_{S'}\|_1 \leq (\frac{2}{\beta} + \frac{1}{3})\|z_{S}\|_1 + (1+\frac{2}{\beta})\|z_{\bar{S}}\|_1 \leq \gamma \|z\|_1,	
\]
for some absolute constant $\gamma < 1$, since $\|z_{S}\|_1 \geq 3 \|z_{\bar{S}}\|_1$.
\end{proof}

We are now ready to proceed with the proof of Theorem~\ref{thm:tail_point_query}.

\begin{proof}[Proof of Theorem~\ref{thm:tail_point_query}]
We analyse the $\Oh(\log n)$ iterations of Algorithm \ref{alg:reed_solomon_point_query}, lines 4-8. Observe that only coordinates in $S$ are modified, and coordinates in $\bar{S}$ remain the same. In the beginning,  $x^{(0)} = x$ and $ H(x^{(0)},k,1) \subseteq S$ by assumption. We apply Lemma \ref{lem:one_invocation} for vectors 
\begin{align*}
r^{(0)} &= x - w^{(0)} = x - \hat{x}^{(0)}\\
 r^{(1)} &= r^{(1)} - w^{(1)}  = x - \hat{x}^{(0)} - \hat{x}^{(1)}\\
&\ldots \\
 r^{(t)} &= r^{(t)} - w^{(t)} = x - \sum_{j=1}^{t-1} \hat{x}^{(j)},
\end{align*}
till
\[	
	\|r^{(t)}_{S}\|_1 < 3 \|(r^{(t)})_{\bar{S}}\|_1 = \|x_{\bar{S}}\|_1.	
\]

Moreover, since $H(x,k,1) \subseteq S$, it holds that 
\[ 
\|x_{\bar{S}}\|_1 \leq \|x_{-k}\|_1 + k \frac{1}{k} \|x_{-k}\|_1 = 2 \|x_{-k}\|_1. 
\]

By Lemma \ref{lem:one_invocation}, the number of iterations till this happens is $\log_{1/\gamma}n < C \log n = R$ for large enough $C$. For any $t' > t$, we claim that it holds
\[		
	\|r^{(t')_{S}}\|_1 < \left(3 + \frac{20}{\beta}\right)\|x_{\bar{S}}\|_1. 
\]

To see this, observe that $r^{(t')}_{S}$ is increased only when $\|r^{(t')}_{S}\|_1 < 3 \|r^{(t')}_{\bar{S}}\|_1$, and the increment is 

\[ 5k \cdot \frac{1}{\beta k}\|r^{(t')}_{S}\|_1 =\frac{5}{\beta} ( \|r^{(t')}_{S}\|_1 + \|(r^{(t')})_{\bar{S}}\|_1) \leq \frac{5}{\beta}\cdot 4 \|x_{\bar{S}}\|_1	\] since at most $5k$ coordinates in $S$ are updated. Thus, we have that

\begin{align*}
\|x - w^{(R+1)}\|_1 &= \left\| x - \sum_{t=0}^{R} \hat{x}^{(t)}\right\|_1\\
& = \left\|\left(x - \sum_{t=0}^{R-1} \hat{x}^{(t)}\right) - \hat{x}^{(R)}\right\|_1 \\
&\leq\frac{1}{\beta k } \|r^{(R)} - \hat{x}^{(R)}\|_1 \\
&= \frac{1}{\beta k} \|r^{(R)}\|_1 \leq \frac{1}{\beta k } ( 3+ \frac{4}{\beta}) \|x_{-k}\|_1 \\
&\leq \frac{1}{k}\|x_{-k}\|_1.\qedhere
\end{align*}
\end{proof}

\subsection{Sublinear-time Decoding}\label{sec:sublinear_explicit}

We give our results in the strict turnstile model. Again, we set $ k = \lceil 1/\epsilon \rceil$. We prove the following theorem, which converts $\ell_{\infty}/\ell_1$ sketches that run in $\Omega(n)$ time to $\ell_{\infty}/\ell_1$ sketches that run in $o(n)$ time. 

\begin{theorem} \label{thm:conversion}
Let $M_{n,k}$ be a family of matrices parametrized by $n,k$, such that given $y= M_{n,k}x$ for $x \in \mathbb{R}^n$, and a set $S$ such that $H(x,k,1) \subseteq S$, it is possible to find an $O(k)$-sparse vector $x'$ such that
	\[	\|x - x'\|_{\infty} \leq \frac{1}{k} \|x_{-k}\|_1,	\]
in time $T(n,k,|S|)$. 
Then there exists a matrix $M$ of 
\[	
	\sum_{i=1}^{\log \log_k n} 2^i \cdot M(n^{1/2^i},k),
\]
rows (where $M(n,k)$ denotes the number of rows of $M_{n,k}$) and an algorithm which, given $y=Mx$, finds in time $\sum_{i=1}^{\log \log_k n} 2^i \cdot T(n^{1/2^i},k,25k^2)$ an $O(k)$ sparse vector $x'$ such that 
\[	
	\|x' - x\|_{\infty} \leq \frac{1}{k} \|x_{-k}\|_1.
\]
\end{theorem}

\begin{proof}
For $ i \in [n]$ we define $\operatorname{first}(i)$ to be the number represented by the first $n_1 := \lceil \frac{1}{2} \log n \rceil$ bits of $i$, and $\operatorname{sec}(i)$ be the number represented by the last $n_2 := \log n - \lceil \frac{1}{2} \log n \rceil$ bits of $i$. This induces an injective map $\pi: [n]\to [N_1]\times [N_2]$, where $N_1 = 2^{n_1}$ and $N_2 = 2^{n_2}$.

For every $x \in \mathbb{R}^n$ we form vectors $v \in \R^{N_1}$ and $u \in \R^{N_2}$ as
\[
v_j = \sum_{ i \in [n]: \operatorname{first}(i) = j} x_i
\] 
and 
\[
	u_j = \sum_{i \in [n]: \operatorname{sec}(i) = j} x_i.
\]
By the assumption of matrix family $M_{n,k}$, we can find vectors $v'\in\R^{N_1}$ and $u'\in \R^{N_2}$ of support size $5k$ such that $\supp(v')\supseteq H(v,k,1)$ and $\supp(u')\supseteq H(u,k,1)$.

Observe that
\[
\|u\| = \|v\|_1 = \|x\|_1, \quad \|u_{-k}\|_1,\|v_{-k}\|_1 \leq \|x_{-k}\|_1.	
\]
by the assumption of the strict turnstile model. Moreover, we have that for all $j \in [N_1]$ such that that there exists $i \in H(x,k,1) \cap \operatorname{first}^{-1}(j)$, it holds that $ |v_j| \geq \frac{1}{k}\|x_{-k}\|_1 \geq \frac{1}{k}\|v_{-k}\|_1$, which implies if we recover the heavy hitters in $v$, we can recover the first $n_1$ bits of heavy hitter indices in $x$. Similarly, if we recover the heavy hitters in $u$, we can recover the last $n_2$ bits of heavy hitter indices in $x$. Formally we have that $H(x,k,1)\subseteq \pi^{-1}(H(v,k,1)\times H(u,k,1))$.

Let $S_1 = \supp(v')$ and $S_2 = \supp(u')$. We have shown that $H(x,k,1)\subseteq \pi^{-1}(S_1\times S_2)$. For every $(i,j) \in S_1 \times S_2$, we query $x_{\pi^{-1}(i,j)}$ using the matrix $M_{n,k}$, and find a vector $x'$ of $5k$ nonzero coordinates such that $\|x'-x\|_1\leq (1/k)\|x_{-k}\|_1$, which we finally output. 

The running time is
\[ 
	T(n,k,k) = 2 T(\sqrt{n},k,\sqrt{n}) + T(n,k,25k^2), 
\]
and number of rows is
\[ 
	2 M(\sqrt{n},k) + M(n,k).
\]

We apply the same idea recursively to $u$ and $v$, each of length $\sqrt{n}$. We stop at the $i$th level of recursion, when $n^{1/2^i} \leq 25k^2$. One can imagine the recursion tree as a binary tree of $\log_k n$ nodes, the $i$th level of which uses $2^iM(n^{1/2^i},k)$ rows.
\end{proof}

We are now ready to prove Theorem~\ref{thm:strict_turnstile}. We restate it with $k = \lceil 1/\epsilon\rceil$ as follows.
\begin{reptheorem}{thm:strict_turnstile}[rephrased]
There exists a strongly explicit matrix $M$ with $O(k^2 (\log_k n)^3)$ rows, which, given $y = Mx$, allows us to find an $O(k)$-sparse vector $\hat x$ such that
	\[
		\|x-\hat x\|_{\infty} \leq \frac{1}{k} \|x_{-k}\|_1,
	\]
in time $O(k^3 (\log_k n)^3)$. 
\end{reptheorem}

\begin{proof}
The proof follows by invoking Theorem~\ref{thm:conversion}, where the good family of matrices is guaranteed by Theorem~\ref{thm:tail_point_query}.  Observe that the number of rows is upper bounded by
\[
	\Oh\left(\sum_{i=1}^{\log \log_k n} 2^i \cdot k^2 \left(\frac{ \log (n/2^i)}{ \log \log n + \log k} \right)^2\right) =
	\Oh\left(k^2\sum_{i=1}^{\log \log_k n} 2^i\left(\frac{\log n}{\log k}\right)^2\right) =
	\Oh\left(k^2 (\log_k n)^3\right)
\]
and the total running time is 
\[
	\Oh\left(2^{\log \log_k n}\cdot k^3 \left(\frac{\log k}{ \log \log k + \log k}\right) + \!\!\!\!\sum_{i=1}^{\log \log_k n -1}\!\!\!\! 2^i \cdot k^3 \left(\frac{\log (n/2^i)}{\log \log (n/2^i) + \log k}\right)^2\right) = \Oh\left(	k^3 (\log_k n)^3 \right).
\]
\end{proof}

\section{Conclusion and Open Problems}

In this work, we present the first algorithm for finding $\ell_1$ heavy hitters ($\ell_{\infty}/\ell_1$ guarantee) deterministically in sublinear time, up to an $\Oh(\log^{\ast}(1/\epsilon))$ factor in the number of measurement from the best superlinear-time algorithm. It still remains to improve the dependence on $\epsilon$ in the running time, ideally to $\Oh(\epsilon^{-2} \poly(\log n))$. The problem could first be approached in the strict turnstile model, where it is possible to avoid the heavy machinery of list-recoverable codes or the clustering algorithm of \cite{LNNT}. Another open problem is to find (fully) explicit constructions that allows a sublinear-time decoding with the number of rows near $\Oh(\epsilon^{-2}\log n)$ in the strict turnstile model. In the general turnstile model, our current understanding and techniques suggest that  an explicit scheme would require an explicit construction of expanders or lossless condensers, together with list-recoverable codes with nearly optimal encoding and decoding time, 
constructions that are currently out of reach. In conclusion, we hope that our work will ignite further work in the field, and towards the resolution of some of these questions.

\bibliographystyle{plain}
\bibliography{biblio} 

\appendix
\section{Encoding One Bit In Two Measurements}\label{sec:example1}

Suppose that $n = 8$ and there are two buckets $\{1,3,4,6\}$ and $\{2,5,7,8\}$. This corresponds to the following $2\times 8$ measurement matrix:
\[
\begin{pmatrix}
1 & 0 & 1 & 1 & 0 & 1 & 0 & 0\\
0 & 1 & 0 & 0 & 1 & 0 & 1 & 1
\end{pmatrix}.
\]
Suppose that for each position we wish to embed a message of two bits, listed as column vectors below.
\[
\begin{pmatrix}
\mathbf{m}_1 & \mathbf{m}_2 & \cdots & \mathbf{m}_8
\end{pmatrix}
=
\begin{pmatrix}
1 & 0 & 0 & 0 & 1 & 1 & 1 & 1\\
1 & 0 & 1 & 1 & 0 & 0 & 1 & 0
\end{pmatrix}.
\]
The measurement matrix will therefore be $8\times 8$ as below.
\[
\Phi = 
\left(\begin{array}{cccccccc}
0 & 0 & 1 & 1 & 0 & 0 & 0 & 0\\
1 & 0 & 0 & 0 & 0 & 1 & 0 & 0\\ 
0 & 0 & 0 & 0 & 0 & 1 & 0 & 0\\
1 & 0 & 1 & 1 & 0 & 0 & 0 & 0\\
\hdashline[2pt/2pt]
0 & 1 & 0 & 0 & 0 & 0 & 0 & 0\\
0 & 0 & 0 & 0 & 1 & 0 & 1 & 1\\
0 & 1 & 0 & 0 & 1 & 0 & 0 & 1\\
0 & 0 & 0 & 0 & 0 & 0 & 1 & 0
\end{array}
\right)
\]
The first $4$ rows correspond to the first bucket and the last $4$ rows the second bucket. For instance, $(\begin{smallmatrix}\Phi_{4,3} \\ \Phi_{4,4}\end{smallmatrix}) = (\begin{smallmatrix} 0\\ 1\end{smallmatrix})$ because the second bit of $\mathbf{m}_4$ is $1$, and we rewrite a single entry $a$ as $(\begin{smallmatrix} 0\\ a\end{smallmatrix})$ for bit $1$; $(\begin{smallmatrix}\Phi_{7,5} \\ \Phi_{7,6}\end{smallmatrix}) = (\begin{smallmatrix} 0\\ 1\end{smallmatrix})$ because the first bit of $\mathbf{m}_7$ is $1$.

To decode, suppose that the signal
\[
x = \begin{pmatrix} 10.1 & -0.1 & 0.3 & 0.2 & -9.7 & 0.1 & 0.2 & -0.2 \end{pmatrix}^T
\]
has two heavy hitters $x_1$ and $x_5$, which are isolated in the respective bucket. Then the measurements are
\[
y = \Phi x = \begin{pmatrix} 0.5 & 10.2 & 0.1 & 10.6 & -0.1 & -9.7 & -10 & 0.2\end{pmatrix}^T
\]
Converting two consecutive measurements back to a bit in succession according to their magnitude, we can thus recover
\[
\mathbf{m} = \begin{pmatrix} 1 & 1 & 1 & 0 \end{pmatrix}.
\]
Here the first bit is $1$ because $|y_1|<|y_2|$, the last bit is $0$ because $|y_7|>|y_8|$. In the recovered message $\mathbf{m}$, the first two bits $11$ is the message recovered from the first bucket and the last two bits $10$ the message from the second bucket. We see that they are exactly the messages ($\mathbf{m}_1$ and $\mathbf{m}_5$) associated with the heavy hitter ($x_1$ and $x_5$) in the respective bucket.
\section{Proof of Theorem \ref{thm:weakerell_infty/ell_1}}\label{sec:weaker_ellinfty/ell1_proof}

Without loss of generality, assume that $\|x_{-k}\|_1 = 1$. We shall apply \Cref{lem:weak} repeatedly to obtain a sequence of vectors $\hat{x}^{(i,r)}$, which admit decompositions $x^{(i,r)} = x - \hat{y}^{(i,r)} - \hat{z}^{(i,r)}$. Consider the following loop invariants, parametrized by $(i,r)$, at the beginning of the $i$-th step in the $r$-th round:

\begin{equation}\label{eqn:loop_1}
\begin{aligned}
|\supp(\hat{y}^{(i,r)})| &\leq s_{i,r} := (i+1)^2 k_r^{2^{-i}}\\
\|\hat{z}^{(i,r)}\|_1 &\leq 2- s_{i,r}/k
\end{aligned}
\end{equation}

We claim that the loop invariants above hold for $(0,r)$ for $r = O(\log^\ast k)$. The base case is $(0,0)$ and the loop invariants holds trivially. Suppose that the loop invariants hold for $(i,r)$, we shall show that it holds for $(i,r+1)$ whenever $r \leq r_0$ for some $r_0 = O(\log^\ast k)$.

To prove the inductive step w.r.t.\ $r$, we consider an inductive proof w.r.t.\ $i$ for a fixed $r$. For the inductive step, if $i < i_r^\ast$ we apply \Cref{lem:weak} with $w_{i,r} = (i+1)^2$ and $s_{i,r} = (i+1)^2 k_r^{2^{-i}}$. We then get that 
\[
|\supp(\hat{y}^{(i,r)})| \leq\sqrt{s_{i,r} w_{i,r}} = (i+1)^2 k_r^{2^{-(i+1)}}\leq s_{i+1,r},
\]
and
\[
 \|\hat{z}^{(i,r)}\|_1 \leq 2-s_{i,r}/(2k)\leq 2-s_{i+1,r}/k
\] 
when $s_{i+1,r}\leq s_{i,r}/2$, that is, when $4(1+\frac{1}{i+1})^4\leq k_r^{2^{-i}}$. 

This proves the loop invariants \eqref{eqn:loop_1} when $k_r^{2^{-i}}\geq \max \{(i+1)^2, 4(1+\frac{1}{i+1})^4 \}$, and that is $i \leq i_r^\ast$ for some $i_r^\ast = \Oh(\log \log k_r)$. At this stage, the residual admits the decomposition $y^{(i_r^\ast,r)} +z^{(i_r^\ast,r)}$ with $|\supp(y^{(i_r^\ast,r)})| \leq \Oh(\log^4 \log k)$ and $\|z^{(i,r)}\|_1\leq 2 - s_{{i_r^\ast},r}/k$. 
%

Now we change our choice of parameters and the loop invariants. In the $i$-th step ($i\geq i_r^\ast + 1$), we claim the following invariants hold at the beginning of the $i$-th step by changing $w_{i,r}$ to $w_{i,r}=1$:
\begin{equation}\label{eqn:loop_2}
\begin{aligned}
|\supp(\hat{y}^{(i,r)})| &\leq s_{i,r} := (s_{i_r^\ast,r})^{2^{-(i-i_r^\ast-1)}}\\
\|\hat{z}^{(i,r)}\|_1 &\leq 2 - s_{i,r}/k
\end{aligned}
\end{equation}
By our choice of $i_r^\ast$ and the argument above the invariants hold when $i = i_r^\ast + 1$. Applying \Cref{lem:weak} with $w_{i,r} = 1$, we see that
\[
|\supp(\hat{y}^{(i,r)})| \leq\sqrt{s_{i,r}} \leq s_{i+1,r}
\]
and
\[
 \|\hat{z}^{(i,r)}\|_1 \leq 2-s_{i,r}/(2k_r)\leq 2-s_{i+1,r}/k,
\]
whenever $s_{i,r}\geq 4$. This proves the loop invariants \eqref{eqn:loop_2} when $s_{i,r}\geq \max\{4, \log\log k_r\}$, which holds when $i\leq i_r^\ast + i_r^+$ for some $i_r^+ = O(1)$ (recall that $s_{i_r^\ast} = O(\log^4\log k_r))$. These steps increases $\supp(\hat x)$ by
leaves us a decomposition of the residual as $y^{(i,r)} + z^{(i,r)}$, where $|\supp(y^{(i,r)})| = s_{i,r} \leq \max\{4, \log \log k_r\}$ and $\|z^{(i,r)}\|_1 \leq 2 - s_{i,r}/k$.

Next, we start a new round by setting $k_{r+1} = s_{0,r+1} = s_{{i_r^\ast+{i_r^+}+1},r}$. The loop invariants in \eqref{eqn:loop_1} continue to hold in the base case $i = 0$. This completes proof of the claim that the loop invariants hold for $(0,r+1)$, provided that $k_{r+1} > 4$. Since $k_{r+1} \leq \log \log k_r$ and $k_0 = k$, the loop invariants in \eqref{eqn:loop_1} hold for all $r\leq r_0$ for some $r_0 = O(\log^\ast k)$.

When $k_{r+1} \leq 4$, that is, there are at most $4$ heavy hitter left, we shall recover all of them in one call to the weak system. Setting $w < 1/4$ in \Cref{lem:weak} yields that $|\supp(\hat y)| < 1$; it thus must hold that $|\supp(\hat y)| = 0$, or $\hat y = 0$, which means that all heavy hitters have recovered. This last call recovers $\hat x_{\text{final}}$, which has support size $O(1)$.

\smallskip
\noindent\textbf{Support size of output.} The support size of the output $\hat x$ is upper bounded by
\[
|\supp(\hat x_{\text{final}})| + \sum_{r,i} |\supp(\hat x^{(i,r)})| \leq \Oh(1) + \sum_r \Oh(k_r) = \Oh(k).
\]

\noindent\textbf{Number of rows.} In all rounds except the last round, the number of rows is bounded by
\[
\Oh\left(\sum_{i=0}^{i_r^\ast} \frac{k^2}{(i+1)^2}\log n\right) + \Oh(i_r^+ \cdot k^2\log n) = \Oh(k^2\log n)
\]
and the last round needs $\Oh(k^2\log n)$ rows. The overall number of rows is therefore $m = \Oh(k^2\log n\log^\ast k)$ as there are $\Oh(\log^\ast k)$ rounds.

\smallskip
\noindent\textbf{Runtime.} Each call to the weak system runs in $\Oh(k^6\poly(\log n))$ time and there are $(\sum_r (i_r^\ast+i_r^+))+1 = \Oh(\log \log k\cdot \log^\ast k)$ calls. Each update of $y^{(i+1,r)}\gets y^{(i,r)}-\Phi \hat x^{(i,r)}$ takes $\Oh(mk)$ since $\Phi$ has $m$ rows and $|\supp(\hat x)| = \Oh(k)$; there are $O(\log \log k\cdot \log^\ast k)$ such updates. The overall runtime is therefore $\Oh(k^6\poly(\log n))$.

\smallskip
\noindent\textbf{Storage of the sketching matrix.} Each weak system uses $\Oh(k\log n)$ random $\Oh(k)$-wise independent hash function and needs space $\Oh(k^2\log n)$ words. We have $\Oh(\log \log k\cdot \log^\ast k)$ such hash functions and thus the total storage for sketching matrix is $\Oh(k^2\log n \cdot \log \log k \log^{\ast} k)$ words.

\end{document}